\documentclass[submission,copyright,creativecommons]{eptcs}

\usepackage{iftex}
\usepackage{graphicx}
\usepackage{setspace}
\usepackage{amsmath,amssymb}
\usepackage{proof, amssymb, my-macro2}
\usepackage{amsthm}
\usepackage{enumitem}

\ifpdf
  \usepackage{underscore}         
  \usepackage[T1]{fontenc}        
\else
  \usepackage{breakurl}           
\fi

\newtheorem{definition}{Definition}
\newtheorem{lemma}{Lemma}
\newtheorem{theorem}{Theorem}
\title{On the Reconstruction of Orthomodular Logic via Restricted Negation}
\author{Tomoaki Kawano
\institute{Faculty of Informatics\\
Kanagawa University\\
Yokohama, Japan}
\email{tomoakikawano.tk@gmail.com}
\and
Ryo Kashima
\institute{Department of Mathematical and Computing Science\\
Institute of Science Tokyo\\
Tokyo, Japan}
\email{kashima@is.titech.ac.jp}
}
\def\titlerunning{On the Reconstruction of Orthomodular Logic via Restricted Negation}
\def\authorrunning{T. Kawano \& R. Kashima}
\begin{document}
\maketitle

\begin{abstract}
Orthomodular logic has long been studied as a fundamental logical framework within quantum logic. In addition to its lattice-theoretic foundations, various formulations have been proposed that incorporate specialized implication connectives and more complex logical connectives.
In this study, we reconstruct orthomodular logic by introducing a connective referred to as restricted negation, which has received comparatively little attention in the existing literature. We demonstrate that this connective provides a natural and expressive basis for formulating the core principles of orthomodular logic. Furthermore, as a main result, we construct a new proof system for orthomodular logic founded on this connective.
\end{abstract}

\section{Introduction}

{\it Quantum logic} (QL) is a logical framework that studies the distinctive properties of propositions arising in quantum mechanics and quantum information science. 
In contrast to classical propositions, quantum mechanical statements exhibit non-classical features such as incompatibility and the failure of distributive laws, which make classical logic insufficient as a foundational framework for quantum theory. 
Research in quantum logic has proceeded in various directions, including the development of proof systems, the analysis of algebraic and lattice-theoretic models, applications to quantum information science, and investigations into their semantic and structural properties \cite{bat2,birk,HANDQ,HANDQS}.  

Among the systems studied in this context, {\it orthomodular logic} (OML) plays a fundamental role.
Quantum states are represented as vectors in a Hilbert space, while observational propositions are associated with closed subspaces of it.
OML employs an {\it orthomodular lattice}, an abstraction of the lattice of closed subspaces, as its model.
A notable difficulty of OML is that it employs the same logical symbols as classical logic, such as conjunction, disjunction, and negation, while assigning them meanings that differ in subtle but essential ways. This similarity in notation, combined with a divergence in behavior, often complicates both the interpretation and the formal analysis of the logic \cite{birk,fag,fin}.

In this study, we investigate a particular logical connective in OML that is motivated by quantum-mechanical considerations. 
Here, we consider connectives constructed by combining the basic connectives {\it conjunction} $\wedge$ and {\it negation} $\sim$.
Intuitively, each propositional variable represents a closed subspace of a Hilbert space, while $\wedge$ yields the intersection subspace and $\sim$ yields the orthogonal complement.
Although several quantum-logical connectives derived from these connectives have been investigated, most existing studies focus on
implication-like connectives such as the {\it Sasaki hook} $\sim\!\!(A \wedge\!\sim\!\!(A \wedge B))$ or projection-like
connectives such as the {\it Sasaki projection} $A \wedge\!\sim\!\!(A \wedge\!\sim\!\!B)$ \cite{fin,har,kotas,ozawa2,ozawa1,sasaki}.
These connectives are associated with physically important meanings.
For example, Sasaki hook can be interpreted as ``$B$ holds after projection onto the subspace in which $A$ is true.''
In this study, we focus on {\it restricted negation} $A \wedge \sim\!\!(A \wedge B)$, a connective that has received little attention.
It has not been systematically studied as a primitive connective for OML.
In particular, its expressive power, its relation to the orthomodular law,
and its proof-theoretic behavior remain unclear.
Although restricted negation is formally reminiscent of the Sasaki hook and related connectives, its significance arises from a different point of view on Hilbert spaces.
Intuitively, this connective corresponds to taking the orthogonal complement of subspace $B$ within the subspace $A$.

The main contributions of this study are as follows. First, we clarify
the algebraic and Hilbert-space interpretation of restricted negation.
In particular, we show that this connective behaves as a relative
orthogonal complement not only in Hilbert-space models but also in
orthomodular lattices. Second, we show that conjunction and negation are
definable from restricted negation together with a constant. Hence the
language of OML can be reconstructed from a single primitive connective.
Third, we show that 
the orthomodular law admits a particularly simple formulation in the language based on restricted negation. 
Finally, we construct a proof system based on restricted negation, in which the orthomodular structure is naturally
incorporated into the inference rules.


Proof systems for OML have been extensively studied \cite{uwe,fag,Fazio2023,tom3,TKMB,Kornell2024,nis1,nis2,samb}.
Each of these approaches has its own distinctive advantages, and no single approach has become predominant.
In ordinary OML, at least two logical connectives, conjunction and negation, are required. However, since both of these connectives can be expressed in terms of restricted negation and the constant, all propositions can be handled using a single connective.
This places the present work in the broader tradition of single-connective formulations, such as classical logic based on the Sheffer stroke. Related approaches have also been studied in QL using the Sheffer stroke \cite{Chajda,McCune}. Like restricted negation, the Sheffer stroke provides a binary operator from which the usual lattice operators can be recovered. The motivations of the two approaches, however, are different. McCune et al. \cite{McCune} employ the Sheffer stroke primarily to obtain compact equational bases and single axioms, whereas restricted negation is motivated by relative orthogonal complementation and is treated here as a primitive logical connective. Our main concern is therefore its logical interpretation, its relation to the orthomodular law, and the development of a sound and complete proof system based on it.
There are also previous studies of OML based on a single connective, especially implication together with constants \cite{OBrunet,fin,har}. In contrast, the present approach centers on a negation-like connective and thus provides a different perspective on the structure and proof theory of OML.
Our new system allows the concept of orthomodular law to be naturally incorporated into the inference rules. 
Orthomodular law has a rather complicated form when formulated using conjunction and negation, whereas it can be represented in a simple form when expressed in terms of 
restricted negation.
Thus, restricted negation is not merely an additional derived connective;
it provides an alternative primitive through which the structure of OML
can be reconstructed.
Since the {\it decidability} of the validity problem for OML has not yet been established, the development of proof systems from diverse perspectives is an important task.

In Section \ref{OMbasic}, we review the basic notions of OML. In Section \ref{RN}, we analyze the fundamental properties of restricted negation. In Section \ref{SCRN}, we construct a new proof system using restricted negation.

\section{Orthomodular logic}\label{OMbasic}

\begin{definition}[Language $\mathcal{L}_1$ and formulas of OML]  
\end{definition}

The language $\mathcal{L}_1$ of OML consists of a countable infinite set of propositional variables $p, q, r,\ldots$, and the logical connectives $\sim, \wedge$. 
Formulas of OML (for $\mathcal{L}_1$) are defined as follows:

\begin{description}
\item$A :: = p |\mkern2mu \sim\!A |\mkern2mu A \wedge A$
\end{description}

For the sake of simplicity, disjunction and implication (Sasaki hook) are defined as abbreviations: $A \vee B$ = $\sim\!(\sim\!A \wedge \sim\!B)$, $A \to_S B =\mkern2mu \sim\!(A \wedge \sim\!(A \wedge B))$.

In what follows, we review two well-known semantics for OML. Despite their equivalence, the two semantics offer different technical advantages, and the most effective choice depends on the theorem to be proved. 
(For Theorem \ref{subAlattice}, the lattice representation is especially convenient, while for the results surrounding Theorem \ref{compOMRp}, the Kripke model provides a more suitable framework.)

\begin{definition}[Ortho lattices and orthomodular lattices]
\end{definition}

An {\it ortho lattice} is a complemented lattice $(L, \sqcap, \sqcup, ', \bf{1}, \bf{0})$ that 
satisfies the following conditions for all $a, b \in L$:

\begin{enumerate}

\item $\exists ! a'$,

\item $a'' = a$,

\item if $a \sqsubseteq b $, then $b' \sqsubseteq a'$,

\item $a \sqcup b = (a' \sqcap b')'$,

\item $a \sqcap a' = \bf{0}$, 

\item $a \sqcup a' = \bf{1}$.

\end{enumerate}

\noi
where $\exists !$ denotes unique existence, and $a \sqsubseteq b$ is defined as $a \sqcap b = a$ (or, equivalently, $a \sqcup b = b$).

An {\it orthomodular lattice} is defined as an ortho lattice satisfying the following {\it orthomodular law}:

\begin{description}

\item
$\forall a, b \in L$, $a \sqcap ( a' \sqcup (a \sqcap b)) \sqsubseteq  b $
\ \ \ \  (or, equivalently,
$a \sqcap ( a \sqcap (a \sqcap b)')' \sqsubseteq  b$) 

\end{description}

The orthomodular law admits the following equivalent formulations.

\begin{description}

\item
$\forall a, b \in L$, $a \sqcap ( a \sqcap (a \sqcap b)')'  = a \sqcap  b$

\end{description}

Intuitively, an orthomodular lattice can be regarded as an abstraction of a Hilbert space.
Elements of $\mathcal{L}$ represent closed subspaces of a Hilbert space, $\bf{1}$ represents the entire space, $\bf{0}$ represents the space consisting only of the origin, $'$ represents the orthogonal complement, $\sqcap$ represents the intersection of two subspaces, and $\sqcup$ represents the subspace spanned by two subspaces.
One of the main features of orthomodular lattice is that the {\it distributive law} 
$a \sqcap (b \sqcup c) = (a \sqcap b) \sqcup (a \sqcap c)$ does not hold.
For example, when we consider a two-dimensional complex Hilbert space, 
the set of all one-dimensional closed subspaces, $\bf 0$, and ${\bf 1}$ forms an orthomodular lattice; however, because of $C_1(|0 \rangle + |1 \rangle) \sqcap (C_2|0 \rangle \sqcup C_3|1 \rangle) = C_1(|0 \rangle + |1 \rangle) \sqcap {\bf 1} = C_1(|0 \rangle + |1 \rangle)$ and $(C_1(|0 \rangle + |1 \rangle) \sqcap C_2|0 \rangle) \sqcup (C_1(|0 \rangle + |1 \rangle) \sqcap C_3|1 \rangle) = {\bf 0} \sqcup {\bf 0} = {\bf 0}$, the distributive law  does not hold.

\begin{definition}[Valuations and validity]
\end{definition}

A {\it valuation} $v$ is a function from the set of propositional variables to $L$ of $(L, \sqcap, \sqcup, ', \bf{1}, \bf{0})$. Then, we extend the definition of $v$ inductively as follows:

\begin{description}

\item $v(\sim A) = v(A)'$,

\item $v(A \wedge B) = v(A) \sqcap v(B)$.

\end{description}

We say that $A$ is {\it valid} in orthomodular lattices if $v(A) = {\bf 1}$ for all $v$ and all orthomodular lattices.

\begin{definition}[O-frame]
\end{definition}
An {\it O-frame} is a tuple $(W, \perp)$ where, 

\begin{description}

\item $W$ is a nonempty set. \  

\item $\bot$ is a binary relation on $W$ which is irreflexive and symmetric. 

\end{description}

Given $X \subseteq W$, $X^{\bot} \stackrel{\mathrm{def}}{=} \{ w \in W |  \forall x \in X$, $ x \bot w \}$. We say that $X$ is {\it $\bot$-closed} or {\it testable} if $X^{\bot\bot} = X$.
Intuitively, each element of $W$ represents a one-dimensional closed subspace (quantum states), $\bot$ represents the orthogonality relation between states, and a $\bot$-closed set represents a closed subspace.
Although the symbol $\bot$ appearing here coincides with the symbol for contradiction, it is used here to denote orthogonality. This notation is standard in the tradition of OML.

\begin{definition}[OM-model and validity]
\end{definition}
An {\it OM-model} is a tuple $(W, \perp, \sigma, V)$ where, 
\begin{list}{}{
  \setlength{\leftmargin}{8pt}
}
\item $(W, \perp)$ is an O-frame.
\item $\sigma$ is a set of $\bot$-closed sets that satisfies the following conditions:
\begin{description}
\item $W \in \sigma, \emptyset \in \sigma$,
\item closed under $\cap$ and $^{\bot}$,
\item for all $X, Y \in \sigma$,
$X \cap (X \cap (X \cap Y)^{\bot})^{\bot} = X \cap Y$ \ (orthomodular law).
\end{description} 
\item $V$ is a function from the set of propositional variables to $\sigma$. Then, we extend the definition of $V$ inductively as follows:
\begin{description}
\item $V(\sim A) = V(A)^{\bot}$,
\item $V(A \wedge B) =  V(A) \cap V(B)$.
\end{description}
\end{list}

We say that $A$ is {\it true} at $x \in W$ if $x \in V(A)$ and write it as $x \models A$.
We say that $A$ is {\it false} at $x \in W$ if $A$ is not true at $x$.
We say that $A$ is {\it valid} in OM-models if $V(A) = W$ for all OM-models.

For readability, we shall use the following notation throughout: $v(A) = v(B)$ and $V(A) = V(B)$ will be denoted by $A = B$, and $v(A) \sqsubseteq v(B)$ and $V(A) \subseteq V(B) $ by $A \leq B$.

\begin{theorem}[\cite{HANDQ}]
For every OM-model, if $V(p)$ is $\bot$-closed for every propositional variable $p,q,r,\ldots$, 
then for every formula $A$, $V(A)$ is $\bot$-closed.
\end{theorem}

\begin{theorem}[\cite{HANDQ}]
For every formula $A$, $A$ is valid in orthomodular lattices iff $A$ is valid in OM-models.
\end{theorem}

These theorems allow us to formulate OML using either 
orthomodular lattices or OM-models. 
As in the case of algebraic and Kripke models for other logics, the respective operators can also be straightforwardly associated ($\sqcap$ is associated with $\cap$ and $'$ is associated with $^{\bot}$) \cite{HANDQ}.

Next, we review a proof system for OML.
Since OML is well suited to {\it sequent calculus}, we adopt this formulation.
Among the various sequent calculi, 
{\bf GOM} \cite{nis1} is the most amenable to modification; therefore, we adopt it as the basis of this study.
A {\it sequent} is defined by $\GA \Rightarrow \DE$ where $\GA$ and $\DE$ are finite sets of formulas.
We define that sequent $\GA \Rightarrow \DE$ is {\it false} at $x \in W$ if all $A \in \GA$ are true at $x$
and all $A \in \DE$ are false at $x$.
If $\GA \Rightarrow \DE$ is not false at $x$, then it is {\it true} at $x$, and write it as $x \models \GA \Rightarrow \DE$.
If $\GA \Rightarrow \DE$ is true at all $x \in W$ in all OM-models, then $\GA \Rightarrow \DE$ is {\it valid} in OM-models.
In the following definition of {\bf GOM}, minor details have been altered from the original for convenience, but the essence remains unchanged.

\ 

\noi
{\bf GOM} \cite{nis1}


\noi
Axiom: 
\begin{center}
$\SEQ{A}{A}$ \qquad 
\end{center}
\noi
Rules: 
$$
\infer[\mbox{(cut)}]{\SEQ{\GA,\PI}{\DE,\SI}}{
  \SEQ{\GA}{\DE, A}
  &
  \SEQ{A, \PI}{\SI}
}
$$
$$
\infer[\mbox{(wL)}]{\SEQ{A,\GA}{\DE}}{
 \SEQ{\GA}{\DE}
}
\hspace{5mm}
\infer[\mbox{(wR)}]{\SEQ{\GA}{\DE,A}}{
 \SEQ{\GA}{\DE}
}
\hspace{5mm}
\infer[\mbox{($\sim$ L)}]{\SEQ{\sim\!A,\GA}{\DE}}{
 \SEQ{\GA}{\DE,A} 
}
\hspace{5mm}
\infer[\mbox{($\sim$ R)}]{\SEQ{\sim \DE}{\ \sim\!A}}
 {\SEQ{A}{\DE}}
$$
$$
\infer[\mbox{($\sim \sim$ L)}]{\SEQ{\sim \sim A,\GA}{\DE}}{
 \SEQ{A,\GA}{\DE} 
}
\hspace{5mm}
\infer[\mbox{($\sim \sim$ R)}]{\SEQ{\GA}{\DE,\sim \sim A}}
 {\SEQ{\GA}{\DE,A}}
$$
$$
\infer[\mbox{(\AND L)}]{\SEQ{A \AND B, \GA}{\DE}}
 {\SEQ{A, \GA}{\DE}}
\hspace{5mm}
\infer[\mbox{(\AND L)}]{\SEQ{A \AND B, \GA}{\DE}}
 {\SEQ{B, \GA}{\DE}}
\hspace{5mm}
\infer[\mbox{(\AND R)}]{\SEQ{\GA}{\DE, A \AND B}}{
 \SEQ{\GA}{\DE, A} & \SEQ{\GA}{\DE, B}
}
$$
$$
\infer[\mbox{(OM)}]{\SEQ{\sim\!A}{\ \sim\!B}}{
 \SEQ{\sim\!B}{\ \sim\!A} & \SEQ{\ \ \sim\!A , B}{}}
$$

\normalsize

Because of the restricted form of ($\sim$R), we need to add the rule of double-negation introduction.
The features of {\bf GOM} are ($\sim$R) and (OM); otherwise, it is largely the same as classical logic.
When $ \GA \Rightarrow \DE$ is derivable in the system 
{\bf S}, we write $\vdash_{\bf S} \GA \Rightarrow \DE$.

\begin{theorem}[\cite{nis1}]
$\vdash_{\bf GOM}  \GA \Rightarrow \DE$ iff $\GA \Rightarrow \DE$ is valid in OM-models.
\end{theorem}

\section{Restricted negation}\label{RN}

Orthomodular law is a central property of OML, but its
conceptual content is not always easy to grasp from its usual formal
definition. It has therefore been studied from various physical,
algebraic, and proof-theoretic viewpoints. In this study, we approach
orthomodular law from the perspective of compound logical connectives.
For instance, orthomodular law can be related to a modus-ponens-like
principle for the Sasaki hook: from $A$ and $\sim\!(A \wedge \sim\!(A \wedge B))$, one obtains
$B$ \cite{bat2,OBrunet,TK1}. 

Our aim is to give another such interpretation, based on a
connective that has received comparatively little attention: restricted
negation.
We write the restricted negation of $B$ by $A$ as $A/B$, and define its semantics by
$A/B = A \wedge\!\sim\!(A \wedge B)$.
Although the term ``negation'' is usually used for unary operators, we use it for this binary operator in view of its intended meaning.
The term ``restricted'' reflects its Hilbert-space interpretation. 
If $A$ is regarded as a closed subspace of a Hilbert space, then $A/B$
represents the orthogonal complement of $B$ (when $B$ is not contained in $A$, we instead consider $A \wedge B$, the part of $B$ lying in $A$)
taken inside the
subspace $A$, rather than in the whole space. 
More precisely, if $H_1$ and $H_2$ are closed subspaces of a Hilbert space $H$ and $H_2 \subseteq H_1$,
then the orthogonal complement of $H_2$ in $H_1$ is given by
$H_2^{\bot(H_1)} = H_1 \cap H_2^{\bot(H)} = H_1 \cap (H_1 \cap H_2)^{\bot(H)}$.
This viewpoint gives a simple reformulation of the orthomodular law.
Using restricted negation, the orthomodular law can be expressed as
$A/(A/B) \leq B$, or equivalently, $A/(A/B) = A \wedge B$.
Indeed, in the usual notation,
$A \wedge \sim\!\!(A \wedge \sim\!\!(A \wedge B)) = A \wedge \sim\!\!( A \wedge (A \wedge \sim\!\!(A \wedge B))) = A/(A/B)$.
Thus, restricted negation allows the orthomodular law to be read as a
form of a double negation principle, not in the whole space, but inside
the restricted subspace $A$.
In particular, when $A$ is the top element $\top$ which represents the whole space, restricted negation coincides with ordinary negation. In this case, the above identity reduces
to the familiar double-negation property
$B = \top \wedge B = \top /(\top /B) = \top \wedge  \sim (  \top \wedge (\top \wedge \sim (\top \wedge B))) = \sim \sim\!B$.
Therefore, orthomodular law can be understood as a generalization of
double negation for relative orthogonal complements.
From the perspective that OML may be viewed as orthologic (OL) supplemented with new notions arising from various interpretations of the orthomodular law, these notions are typically absent from OL itself. In the case of $A/(A/B) = A \wedge B$, however, the notion can be regarded as a generalization of orthogonal closure $\sim \sim\!B = B$, which is already present in OL.
For this reason, the extension using $/$ may be viewed as a more natural extension of OL.

Another distinctive feature of this restricted negation is its connection with dynamic notions.
As discussed in the Introduction, the Sasaki hook embodies a dynamic notion of projection. It can therefore be represented by a formula involving the modal connective $[A]B$.
With this notation, the orthomodular law can be expressed as $\sim [A] \sim [A] B \leq B$.
Under a modal-logical interpretation, this formula can be read as ``if  $x(A)y$ and $y \not\perp z$, then there exists $w \in W$ such that $z(A)w$ and $w \not\perp x$'' and it captures the self-adjointness of the projection operator.
Since restricted negation is, by definition, the negation of the Sasaki hook, $\sim [A] \sim [A] B \leq B$ corresponds to $A/(A/B)\leq B$ in terms of restricted negation.

In the lattice representation, we also introduce a new operator $b^{a}$ for restricted negation as $a \sqcap (a \sqcap b )'$.
Let $L_a = \{b \in L | b \sqsubseteq a\}$ and $b \sqcup_a c = (b^a \sqcap c^a)^a$.
Also in orthomodular lattices, as in the Hilbert space case, the following theorem provides evidence that restricted negation represents the orthogonal complement operation in abstract subspace $L_a$.

\begin{theorem}\label{subAlattice} Let $(L, \sqcap, \sqcup, ', \bf{1}, \bf{0})$ be an orthomodular lattice and let $a \in L$. Then, $(L_a, \sqcap, \sqcup_a, ^{a}, a, \bf{0})$ is an  orthomodular lattice.
\end{theorem}
\begin{proof}
\begin{enumerate}[align=left, leftmargin=*]

\item This follows directly from the definition.

\item $b^{aa} = a \sqcap (a \sqcap (a \sqcap (a \sqcap b)'))' = a \sqcap (a \sqcap (a \sqcap b)')' = a \sqcap b = b$.

\item $b \sqsubseteq c \Rightarrow a \sqcap b \sqsubseteq a \sqcap c \Rightarrow (a \sqcap c)' \sqsubseteq (a \sqcap b)' \Rightarrow  a \sqcap (a \sqcap c)' \sqsubseteq  a \sqcap (a \sqcap b)' \Rightarrow c^{a} \sqsubseteq b^{a}$.

\item From the definition, $b \sqcup_a c = (b^a \sqcap c^a)^a$.

\item $b \sqcap b^{a} = b \sqcap (a \sqcap (a \sqcap b)') = (b \sqcap a) \sqcap (a \sqcap b)' = \bf{0}$.
 
\item $b \sqcup_a b^{a} = a \sqcap ( a \sqcap   ( a \sqcap (a \sqcap b)' \sqcap  b^{aa}))' = a \sqcap ( a \sqcap   ( a \sqcap (a \sqcap b)' \sqcap b ))' = a \sqcap ( a \sqcap {\bf 0})' = a \sqcap {\bf 1} = a$.

\item[(Orthomodularity)]\ $b \sqcap (b \sqcap (b \sqcap c)^{a})^{a} = b \sqcap a \sqcap ( a \sqcap b \sqcap  a \sqcap (a \sqcap b \sqcap c)')' $

$= (b \sqcap a) \sqcap ( (b \sqcap  a) \sqcap ((b \sqcap a) \sqcap c)')' = b \sqcap ( b  \sqcap (b \sqcap c)')' = b \sqcap c$.

\end{enumerate}
\end{proof}

Another advantage of restricted negation is that, when combined with $\top$, it allows all existing connectives $\sim$ and $\wedge$ (and so $\vee$ and $\to_S$) to be expressed in terms of this single binary connective, as shown below:

\begin{description}

\item $\sim\!A = \top \wedge \sim (\top \wedge A) = \top/A$

\item $A \wedge B = A/(A/B)$

\end{description}

On the basis of these considerations, we define a new language for OML.
\begin{definition}[Language $\mathcal{L}_2$ and formulas of OML]
\end{definition}
The language $\mathcal{L}_2$  of OML consists of a countable infinite set of propositional variables $p, q, r,\ldots$, constant symbol $\top$, and the logical connective $/$.
Then, formulas of OML (for $\mathcal{L}_2$) are defined as follows:

\begin{description}

\item $A :: =  \top \ |\  p \ | \ A / A $

\end{description}

This binary connective $/$ is assumed to be right-associative.
Furthermore, since there is only one binary connective, we abbreviate $A/B$ as $AB$.
For example, $A/(B/(C/D)) $ is abbreviated as $ABCD$.
Hereafter, for the purpose of proving completeness of the sequent calculi, we develop the argument based on OM-models rather than the lattices.
As in the previous section, we define a valuation $V$ as a function from the set of propositional variables to $\sigma$. Then, we extend the definition of $V$ inductively as follows:

\begin{description}


\item $V(\top) = W$

\item $V(AB) =  V(A) \cap (V(A) \cap V(B))^{\bot}$ 

\end{description}

The notions of truth and validity are defined in the same way as in the previous section.
In terms of the framework of OM-models, the correspondences mentioned above are given as follows:


\begin{lemma}
For all OM-models,
$V(\top B) = V(B)^{\perp}$, and  $V(AAB) = V(A) \cap V(B) $  

\end{lemma}
\begin{proof}
$V(\top B) = V(\top) \cap ( V(\top) \cap V(B) )^{\perp} = W \cap ( W \cap V(B) )^{\perp} = V(B)^{\perp}$.

$V(AAB) = V(A) \cap (V(A) \cap V(AB))^{\perp} = V(A) \cap (V(A) \cap (V(A) \cap (V(A) \cap V(AB))^{\perp} ))^{\perp} \\
= V(A) \cap (V(A) \cap (V(A) \cap V(B))^{\perp})^{\perp} = V(A) \cap V(B)$. 

The final equality follows from orthomodularity of $\sigma$.
\end{proof}


As in general approaches that analyze logics using only propositional variables together with $\to$ and $\bot$, the present study also requires the constant $\top$, since propositional variables and $/$ alone are insufficient.
Specifically, we have the following theorem.

\begin{theorem}
There exists no valid formula consisting solely of propositional variables and $/$.
\end{theorem}
\begin{proof}
If $A/B$ is valid, then $A$ must be valid because $A/B = A \wedge\!\sim\!(A \wedge B)$.
Applying this to $A$ and iterating the same argument, we would eventually obtain that a single propositional variable is valid. However, this is not the case.
\end{proof}



\section{Sequent calculi with restricted negation}\label{SCRN}

We now construct sequent calculi of OML based on $\mathcal{L}_2$.
We begin by introducing a basic system {\bf OMR$^{-}$}. While {\bf OMR$^{-}$} is sound, its completeness is currently unknown.
Subsequently, we introduce two extensions of {\bf OMR$^{-}$}, namely {\bf OMR$^{S}$} and {\bf OMR$^{\top}$}. Both satisfy completeness and exhibit distinct characteristics.

\ 

\noi
{\bf OMR$^{-}$} 

\noi
Axioms: 
\begin{center}
$\SEQ{A}{A}$ \hspace{1cm} $\SEQ{}{\top}$ \qquad 
\end{center}

\noi
Rules: 
$$
\infer[\mbox{(cut)}]{\SEQ{\GA,\PI}{\DE,\SI}}{
  \SEQ{\GA}{\DE, A}
  &
  \SEQ{A, \PI}{\SI}
}
$$
$$
\infer[\mbox{(wL)}]{\SEQ{A,\GA}{\DE}}{
 \SEQ{\GA}{\DE}
}
\hspace{5mm}
\infer[\mbox{(wR)}]{\SEQ{\GA}{\DE,A}}{
 \SEQ{\GA}{\DE}
}
$$
$$
\infer[\mbox{(\AND L1)}]{\SEQ{AAB, \GA}{\DE}}
 {\SEQ{A, \GA}{\DE}}
\hspace{5mm}
\infer[\mbox{(\AND L2)}]{\SEQ{AAB, \GA}{\DE}}
 {\SEQ{B, \GA}{\DE}}
$$
$$
\infer[\mbox{(\AND R)}]{\SEQ{\GA}{\DE, AAB}}{
 \SEQ{\GA}{\DE, A} & \SEQ{\GA}{\DE, B}
}
$$
$$
\infer[\mbox{(/L)}]{\SEQ{AC,B}{}}{
 \SEQ{B}{C} 
}
\hspace{5mm}
\infer[\mbox{(/R)}]{\SEQ{AC}{AB}}
 {\SEQ{B}{C}}
$$


One may also regard (/L) and (/R) as rules derived from the rules for the ordinary negation $\sim$, where $\sim$ is replaced by the restricted negation $A/$.
As will be shown later, in terms of completeness, (/L) serves only to make the sequent $A, \top A \Rightarrow$ derivable. Hence, the proof-theoretic strength remains unchanged even if the simpler sequent $A, \top A \Rightarrow$ is added as an axiom instead of (/L). We nevertheless adopt (/L)  here for comparison with {\bf GOM}.


\begin{theorem}[Soundness theorem for {\bf OMR$^{-}$}]\label{soundOMR-} If $\GA \Rightarrow \DE$ is provable in {\bf OMR$^{-}$}, then 
$\GA \Rightarrow \DE$ is valid.
\end{theorem}
\begin{proof}
The proof proceeds by induction on the inference rule used in the last step.
The proof for (cut), (wL), and (wR) are analogous to those in the classical case.
The proof for ($\wedge$L1), ($\wedge$L2), and ($\wedge$R) are also analogous to those in the classical case because $AAB = A \wedge B$.

\begin{description}

\item (/L)
Suppose that $B \Rightarrow C$ is valid.
In all OM-models, $V(B) \subseteq V(C)$. Then, $V(A) \cap V(B) \subseteq V(A) \cap V(C)$,
$(V(A) \cap V(B)) \cap (V(A) \cap V(C))^{\bot}  = \emptyset$, and $(V(A) \cap (V(A) \cap V(C))^{\bot}) \cap V(B) = V(AC) \cap V(B) = \emptyset$.
Therefore, $AC, B \Rightarrow$ is valid.

\item (/R)
Suppose that $B \Rightarrow C$ is valid.
In all OM-models, $V(B) \subseteq V(C)$. Then, $V(A) \cap V(B) \subseteq V(A) \cap V(C)$,
$(V(A) \cap V(C))^{\bot} \subseteq (V(A) \cap V(B))^{\bot}$, and $V(A) \cap (V(A) \cap V(C))^{\bot} \subseteq V(A) \cap (V(A) \cap V(B))^{\bot}$.
Therefore, $AC \Rightarrow AB$ is valid.
\end{description}
\end{proof}

We define {\bf OMR$^{S}$} as the system obtained by adding the following axioms 
(S1) and (S2) to {\bf OMR$^{-}$}.
(Here, $S$ stands for ``semantic equivalence axiom''.)

\begin{center}
$AB \Rightarrow AA\top AAB$ \ (S1)  \hspace{1cm} $AA \top AAB \Rightarrow AB$ \ (S2)  \qquad 
\end{center}

Intuitively, (S1) and (S2) correspond to expressing the meaning of $AB$ in terms of $V(AAB) = V(A) \cap V(B)$ and $V(\top B) = V(B)^{\bot}$.
In other words, this can be seen as a ``translation'' of 
$AB$ in 
$\mathcal{L}_2$ into $\mathcal{L}_1$ using $\wedge$ and $\sim$.
Although some rules of {\bf GOM} and those of {\bf OMR$^{S}$} are semantically similar,
there are essential differences between the two systems. By removing the
orthomodularity rule (OM) from {\bf GOM}, one obtains a proof system for the OL. In this sense, {\bf GOM} treats
orthomodularity as an additional structural principle imposed on top of OL.
{\bf OMR$^{S}$} takes a different approach. It has no rule that directly corresponds
to (OM); instead, properties related to orthomodularity are internalized
into the rules for restricted negation. For example, the rules
$(\wedge$ L1), $(\wedge$ L2), and $(\wedge$ R) reflect the property
$AAB = A\wedge B$, which holds in OML but not in OL. 
Moreover, ($\sim \sim$L) and ($\sim \sim$R), which had to be additionally introduced in {\bf GOM}, are also incorporated within $(\wedge$ L1), $(\wedge$ L2), and $(\wedge$ R) because $\top \top A = A$.
Thus, in {\bf OMR$^{S}$}, the
behavior of conjunction, the effect of orthomodularity, and the concept of orthogonal closure are not treated as separate proof-theoretic components, but are combined through the
single connective of restricted negation. This is one of the main
advantages of {\bf OMR$^{S}$}.

This feature also gives {\bf OMR$^{S}$} a more uniform syntactic and proof-theoretic
structure. Since every compound formula is generated from the same
primitive connective, definitions and arguments by induction on formulas
involve fewer primitive cases. Semantic clauses, translations, and proof
transformations can therefore be formulated in a more systematic way.
Moreover, the design of the proof system can be concentrated on the
behavior of restricted negation alone, rather than distributing the
analysis among several independent connectives. In this respect, {\bf OMR$^{S}$}
does not merely reformulate OML with fewer symbols; it organizes the theory around a single connective that simultaneously captures
logical composition and the orthomodular structure.


\begin{theorem}[Soundness theorem for {\bf OMR$^{S}$}]\label{soundS} If $\GA \Rightarrow \DE$ is provable in {\bf OMR$^{S}$}, then 
$\GA \Rightarrow \DE$ is valid.
\end{theorem}
\begin{proof}
Validity of (S1) and (S2) is clear from the definition of $V(AB)$, $V(AAB) = V(A) \cap V(B)$, and $V(\top B) = V(B)^{\bot}$.
Therefore, from Theorem \ref{soundOMR-},  {\bf OMR$^{S}$} is also sound.
\end{proof}

A natural approach to proving the completeness of {\bf OMR$^{S}$} would be to exploit the completeness of {\bf GOM}. Nevertheless, despite the presence of translation-like axioms (S1) and (S2), the differences in language and the significant divergence in the form of the axioms make this approach infeasible at present.
(For example, $A \wedge B$ is a combination of two formulas, whereas $AAB$ is a combination of three; therefore, it is difficult to apply the same structural induction to both $A \wedge B$ and $AAB$.)
From a proof-theoretic standpoint, establishing the result in a system formulated solely in terms of restricted negation is preferable. 
Therefore, it is desirable to provide an independent proof based on the construction of a canonical model.
However, this construction is difficult when working with (S1) and (S2).
Our approach is to develop another system {\bf OMR$^{\top}$}, prove its completeness, and through this, establish the completeness of {\bf OMR$^{S}$}.
We define {\bf OMR$^{\top}$} as the system obtained by adding the following rules to {\bf OMR$^{-}$}.

$$
\infer[\mbox{($\top$1)}]{\SEQ{A, \top C}{AB}}{
 \SEQ{A, B}{C} 
}
\hspace{5mm}
\infer[\mbox{($\top$2)}]{\SEQ{\GA}{\top AB}}
 {\SEQ{\GA}{AAB}}
\hspace{5mm}
\infer[\mbox{($\top$3)}]{\SEQ{AB, \GA}{\DE}}{
 \SEQ{A,\GA}{\DE} 
}
$$

These rules are not particularly well-structured in form and are introduced as tools for establishing completeness; however, each of them can also be given a semantic interpretation.
($\top$1) is similar to (/R), making use of the fact that, on the left-hand side of a sequent, $\{A, B\}$ is interpreted as $A \wedge B$.
($\top$2) can be viewed as analogous to the rule for $\vee$ because $\top AB$ is equivalent to Sasaki hook $\top AB = \sim\!(A \wedge \sim\!(A \wedge B)) = \sim\!A \vee (A \wedge B)$. 
($\top$3) can be viewed as a variant of the rule for $\wedge$ because $AB = A \wedge \sim\!(A \wedge B)$.


\begin{lemma}\label{=Stop}
 $\vdash_{\bf OMR^{S}} \GA \Rightarrow \DE$ iff \ $\vdash_{\bf OMR^{\top}} \GA \Rightarrow \DE$.
\end{lemma}
\begin{proof}
The proof proceeds by showing that any axiom occurring only in {\bf OMR$^{S}$} can be constructed using the rules in {\bf OMR$^{\top}$} and vice versa.
\begin{description}
\item (S1)

$$
\infer[\mbox{(cut)}]{\SEQ{A B}{\top  A A B}}
{
    \infer[\mbox{($\wedge$R)}]{\SEQ{A B}{\top   \top  A B}}    
    {
    {\SEQ{A B}{A B}}    
    {   
    }
    &
    \infer[\mbox{(wL)}]{\SEQ{A B}{\top}}    
    {
    {\SEQ{}{\top}}    
    }  
    }
    &
    \infer[\mbox{(/R)}]{\SEQ{\top   \top  A B}{\top  A A B}}    
    {
    \infer[\mbox{$(\top 2)$}]{\SEQ{A A B}{\top  A B}}    
    {
    {\SEQ{A A B}{A A B}}   
    }
    }
}
$$
$$
\infer[\mbox{$(\wedge R)$}]{\SEQ{A B}{A A  \top  A A B}}
{
    \infer[\mbox{$(\top 3)$}]{\SEQ{A B}{A}}    
    {
    {\SEQ{A}{A}}     
    }
    &
    \infer{\SEQ{A B}{\top  A A B}}    
    {
    \vdots
    }
}
$$

\item (S2)
$$
\infer[\mbox{($\wedge$L1)}]{\SEQ{A A  \top  A A B}{A B}}
{
    \infer[\mbox{($\wedge$L2)}]{\SEQ{A, A A  \top  A A B}{A B}}    
    {
    \infer[\mbox{$(\top 1)$}]{\SEQ{A, \top  A A B}{A B}}    
    {
    \infer[\mbox{($\wedge$R)}]{\SEQ{A, B}{A A B}}    
    {
    \infer[\mbox{(wL)}]{\SEQ{A,B}{A}}    
    {
    {\SEQ{A}{A}}     
    }
    &
    \infer[\mbox{(wL)}]{\SEQ{A,B}{B}}    
    {
    {\SEQ{B}{B}}  
    }  
    }   
    }
    }
}
$$

\item ($\top$1)
$$
\infer[\mbox{(cut)}]{\SEQ{A, \top   C}{A B}}
{
    \infer[\mbox{($\wedge$R)}]{\SEQ{A, \top   C}{A A  \top  A A B}}    
    {
    \infer[\mbox{(wL)}]{\SEQ{A, \top   C}{A}}    
    {
    {\SEQ{A}{A}}     
    }
    &
    \infer[\mbox{(wL)}]{\SEQ{A, \top   C}{\top  A A B}}    
    {
    \infer[\mbox{(/R)}]{\SEQ{\top   C}{\top  A A B}}    
    {
    \infer[\mbox{($\wedge$L1)}]{\SEQ{A A B}{C}}    
    {
    \infer[\mbox{($\wedge$L2)}]{\SEQ{A, A A B}{C}}    
    {
    {\SEQ{A, B}{C}}  
    }
    }
    }
    }
    }
    &
    {\SEQ{A A  \top  A A B}{A B}}    
}
$$

\item ($\top$2)

$$
\infer[\mbox{(/R)}]{\SEQ{\top  \top  A A B}{\top  A B}}    
    {
    \infer[\mbox{(cut)}]{\SEQ{A B}{\top  A A B}}    
    {
    {\SEQ{A B}{AA\top  A A B}}    
    &
    \infer[\mbox{($\wedge$L2)}]{\SEQ{AA\top AAB}{\top  A A B}}    
    {
    {\SEQ{\top A A B}{\top  A A B}} 
    }
}
}
$$

$$
\infer[\mbox{(cut)}]{\SEQ{\GA}{\top   A B}}
{
    {\SEQ{\GA}{ A A B}}    
    &
    \infer[\mbox{(cut)}]{\SEQ{A A B}{\top  A B}}    
    {
    \infer[\mbox{($\wedge$R)}]{\SEQ{A A B}{\top   \top  A A B}}    
    {
    \infer[\mbox{(wL)}]{\SEQ{A A B}{\top}}
    {{\SEQ{}{\top}}  
    }
    &   
    {\SEQ{A A B}{A A B}}  
    }
    &
    \infer{\SEQ{\top  \top  A A B}{\top  A B}}
    {\vdots} 
    }
}
$$

\item ($\top$3)

$$
\infer[\mbox{(cut)}]{\SEQ{A B, \GA}{\DE}}
{
    \infer[\mbox{(cut)}]{\SEQ{A B}{A}}    
    {
    {\SEQ{A B}{A A \top  A A B}}     
    &
    \infer[\mbox{($\wedge$L1)}]{\SEQ{A A \top  A A B}{A}}   
    {
    {\SEQ{A}{A}}   
    }
    }
    &
    {\SEQ{A, \GA}{\DE}}   
}
$$

\normalsize
\end{description}
\end{proof}

\begin{theorem}[Soundness theorem for {\bf OMR$^{\top}$}]\label{soundtop} If $\GA \Rightarrow \DE$ is provable in {\bf OMR$^{\top}$}, then 
$\GA \Rightarrow \DE$ is valid.
\end{theorem}
\begin{proof}
A corollary of Theorem \ref{soundS} and Lemma \ref{=Stop}.
\end{proof}

\begin{lemma} [Disjunction theorem for {\bf OMR$^{-}$}, {\bf OMR$^{S}$}, and {\bf OMR$^{\top}$}]\label{disjOMR-} If $\GA \Rightarrow \DE$ is provable in {\bf OMR$^{-}$} ({\bf OMR$^{S}$}, {\bf OMR$^{\top}$}) and $\DE \neq \emptyset$, then there exists $A \in \DE$ such that 
$\GA \Rightarrow A$ is provable in {\bf OMR$^{-}$} ({\bf OMR$^{S}$}, {\bf OMR$^{\top}$}).
\end{lemma}
\begin{proof}
The proof is almost the same as the one for {\bf GOM} in \cite{nis1}.
\end{proof}

We now define the tools necessary to prove completeness.
An {\it infinite sequent} is defined as $\GA \Rightarrow \DE$, where $\GA$ and $\DE$ are (potentially infinite) sets of formulas. Although we employ the same symbols as for sequent, it should be clear from the context which interpretation is meant.
We define $\GA' \Rightarrow \DE'$ to be a {\it finite subsequent} of $\GA \Rightarrow \DE$ if
$\GA' \subseteq \GA$, $\DE' \subseteq \DE$, and $\GA'$ and $\DE'$ are finite sets.
We say that an infinite sequent is unprovable if all of its finite subsequents are not provable.
The set of all formulas for $\mathcal{L}_2$ is denoted by $\Omega_2$.

\begin{lemma}\label{maxunp}
If $\GA \Rightarrow \DE$ (consisting of $\mathcal{L}_2$) is unprovable in {\bf OMR$^{\top}$},
there exists an unprovable infinite sequent $\GA' \Rightarrow \DE'$ such that $\GA \subseteq \GA'$, $\DE \subseteq \DE'$, and $\GA' \cup \DE' = \Omega_2$. 
\end{lemma}
\begin{proof}
Because (cut) is included in {\bf OMR$^{\top}$}, if $\GA \Rightarrow \DE$ is unprovable, for all $A \in \Omega_2$, at least one of $A, \GA \Rightarrow \DE$ or $\GA \Rightarrow \DE, A$ is unprovable.
Let the formulas in $\Omega_2$ be enumerated as $A_1, A_2,\ldots$.
Starting from $\GA \Rightarrow \DE (= \GA_0 \Rightarrow \DE_0)$,
we construct $\GA_0 \Rightarrow \DE_0, \GA_1 \Rightarrow \DE_1, \GA_2 \Rightarrow \DE_2, \ldots$ by successively adding each formula to the left or right side in such a way that unprovability is preserved. Then, we define $\GA' \Rightarrow \DE'$ as $\underset{n < \omega}{\bigcup}\GA_n \Rightarrow \underset{n < \omega}{\bigcup}\DE_n$.
Since any finite subsequent of $\GA' \Rightarrow \DE'$ contains only finitely many formulas, it is already a subsequent of some $\GA_n \Rightarrow \DE_n$. Each $\GA_n \Rightarrow \DE_n$ is unprovable, and hence so is any such finite subsequent. Thus, by definition, $\GA' \Rightarrow \DE'$  is unprovable.
\end{proof}

We define  $\top \GA = \{ \top A |\mkern2mu A \in \GA \}$.
A canonical model $(W_c, \perp_c, \sigma_c, V_c )$ is defined as follows:

\begin{enumerate}

\item $W_c$ is defined as the set of all infinite sequents $\GA \Rightarrow \DE$ that satisfy the following conditions:
\begin{description}
\item $\GA \Rightarrow \DE$ is unprovable in {\bf OMR$^{\top}$},
\item $\GA \cup \DE = \Omega_2$.
\end{description}

\item $\perp_c \stackrel{\mathrm{def}}{=} \{ (\GA \Rightarrow \DE, \GA' \Rightarrow \DE') | \exists A \in \Omega_2, A \in \GA$ and $\top A \in \GA'$\}.

\item $\sigma_c \stackrel{\mathrm{def}}{=} \{ \|A\| |\mkern2mu A \in \Omega_2 \}$, where $\|A\| = \{\GA \Rightarrow \DE \in W_c | A \in \GA\}$.

\item $V_c(p) \stackrel{\mathrm{def}}{=} \{ \GA \Rightarrow \DE \in W_c |\mkern2mu p \in \GA \}$ .

\end{enumerate}

\begin{lemma}
Let $( \GA \Rightarrow\DE) \in W_c$. 
If a finite sequent $\GA' \Rightarrow \DE'$ is provable, then it is impossible that
$\GA' \subseteq \GA$ and $\DE' \subseteq \DE$. Consequently, if $\GA' \Rightarrow A$ is provable and $\GA' \subseteq \GA$, then $A \in \GA$. If $A \Rightarrow \DE'$ is provable and $\DE' \subseteq \DE$, then $A \in \DE$.
\end{lemma}
\begin{proof}
If $\GA' \subseteq \GA$ and $\DE' \subseteq \DE$, then $\GA' \Rightarrow \DE'$ is a finite sub-sequent of $\GA \Rightarrow \DE$. However, the definition of $W_c$ implies that all such finite subsequents must be unprovable, leading to a contradiction.
\end{proof}

\begin{lemma}\label{pOMm} In $(W_c, \perp_c, \sigma_c, V_c )$, for any $\GA \Rightarrow \DE \in W_c$ and $A \in \Omega_2$, $(\GA \Rightarrow \DE) \models A $ iff $A \in \GA$.
\end{lemma}
\begin{proof}
The proof proceeds by induction on the structure of formulas.
\begin{list}{}{
  \setlength{\leftmargin}{8pt}
}
\item Suppose that $p \in \GA$. From the definition of $V_c$, $(\GA \Rightarrow \DE) \models p$.

\item Suppose that $p \in \DE$. From the definition of $V_c$, $(\GA \Rightarrow \DE) \not\models p$.

\item Suppose that $AB \in \GA$. 
To show that $AB$ is true at $(\GA \Rightarrow \DE) \in W_c$, it suffices to establish 
$(\GA \Rightarrow \DE) \models A$ and $(\GA \Rightarrow \DE)$ is orthogonal 
to all $(\GA' \Rightarrow \DE') \in W_c$ such that $(\GA' \Rightarrow \DE') \models A$
and $(\GA' \Rightarrow \DE') \models B$.

$AB \Rightarrow A$ is provable in {\bf OMR$^{\top}$} as follows:

$$
\infer[\mbox{($\top$3)}]{\SEQ{AB}{A}}{
 \SEQ{A}{A} 
}
$$

Therefore, $A \in \GA$. From the induction hypothesis, $(\GA \Rightarrow \DE) \models A$.


From the induction hypothesis, for all $(\GA' \Rightarrow \DE') \in W_c$, 
if $(\GA' \Rightarrow \DE') \models A$ and $(\GA' \Rightarrow \DE') \models B$,
then $A,B \in \GA'$. 
$A,B \Rightarrow \top AB$ is provable in {\bf OMR$^{\top}$} as follows:

$$
\infer[\mbox{($\top$2)}]{\SEQ{A,B}{\top AB}}
{
  \infer[\mbox{($\wedge$R)}]{\SEQ{A,B}{AAB}}
  {
        \infer[\mbox{(wL)}]{\SEQ{A,B}{A}}    
    {
    {
    {\SEQ{A}{A}}
    }       
    }
    &
    \infer[\mbox{(wL)}]{\SEQ{A,B}{B}}    
    {
    {
    {\SEQ{B}{B}}
    }       
    }
  }
}
$$

Therefore, if $A,B \in \GA'$, then $\top AB \in \GA'$. From the definition of $\perp_c$, $(\GA \Rightarrow \DE)\perp_c(\GA' \Rightarrow \DE')$.
Therefore, $(\GA \Rightarrow \DE)$ is orthogonal 
to all $(\GA' \Rightarrow \DE') \in W_c$ such that $(\GA' \Rightarrow \DE') \models A$
and $(\GA' \Rightarrow \DE') \models B$.

\item Suppose that $AB \in \DE$. 
If $A \in \DE$, from the induction hypothesis, $(\GA \Rightarrow \DE) \not\models A$.
Therefore, $(\GA \Rightarrow \DE) \not\models AB$.

Suppose that $A \in \GA$. 
For the sake of contradiction, suppose that $A, B \Rightarrow \top \GA$ is provable in {\bf OMR$^{\top}$}.
From Lemma \ref{disjOMR-}, there exists $C \in \GA$ such that $A,B\Rightarrow \top  C$ is provable.
Then, $A,C \Rightarrow AB$ is provable as follows:

$$
\infer[\mbox{(cut)}]{\SEQ{A, C}{AB}}
{
    \infer[\mbox{($\wedge$R)}]{\SEQ{C}{\top  \top  C}}    
    {\SEQ{C}{\top}
    &
    {\SEQ{C}{C}}
    }      
    &
    \infer[\mbox{($\top$1)}]{\SEQ{A, \top  \top  C}{AB}}    
    {\SEQ{A, B}{\top  C}}
}
$$

Because $A,C \in \GA$ and $AB \in \DE$, $\GA \Rightarrow \DE$ is provable. This is a contradiction.
Therefore, $A, B \Rightarrow \top  \GA$ is not provable.
Then, from Lemma \ref{maxunp}, there exists an unprovable infinite sequent $\GA' \Rightarrow \DE'$ such that $A, B \in \GA'$, $\top \GA \subseteq \DE'$ and $\GA' \cup \DE' = \Omega_2$.
We have $(\GA' \Rightarrow \DE') \not\perp_c (\GA \Rightarrow \DE)$ because 
$\top \GA \subseteq \DE'$.
From the induction hypothesis, $(\GA' \Rightarrow \DE') \models A$ and $(\GA' \Rightarrow \DE') \models B$.
This means that one of the above conditions for 
$AB$ to be true is violated.
Therefore, $(\GA \Rightarrow \DE) \not\models AB$.
\end{list}
\end{proof}

\begin{lemma}\label{cOMm} $(W_c, \perp_c, \sigma_c, V_c )$ is an OM-model.
\end{lemma}
\begin{proof}
\begin{description}
\item The irreflexivity of $\perp_c$

Because $A, \top A \Rightarrow $ is provable by (/L), for all $(\GA \Rightarrow \DE) \in W_c$ and
$A \in \Omega_2$, if $A \in \GA$, then $\top A \in \DE$. 
Since no formula can belong to both $\GA$ and $\DE$, it follows that $\top A \notin \GA$. Hence 
$(\GA \Rightarrow \DE) \not\perp_c (\GA \Rightarrow \DE)$.

\item The symmetry of $\perp_c$

Suppose that $(\GA \Rightarrow \DE) \perp_c (\GA' \Rightarrow \DE')$ because of $A \in \GA$ and $\top A \in \GA'$.
Since $A \Rightarrow \top \top A$ is provable by (wL) and ($\wedge$R), $\top \top A \in \GA$. Therefore,  $(\GA' \Rightarrow \DE') \perp_c (\GA \Rightarrow \DE)$. 

\item $\bot$-closedness of $V_c (p)$

By the definition of $\bot$-closed set, $X \subseteq X^{\bot\bot}$ holds for all $X \subseteq W$. Therefore, it suffices to prove that 
$V_c(p)^{\bot\bot} \subseteq V_c(p)$.
For the sake of contradiction, suppose that $V_c(p)^{\bot\bot} \not\subseteq V_c(p)$.
That is, there exists $(\GA \Rightarrow \DE, p) \in W_c$
such that $(\GA \Rightarrow \DE, p) \perp_c x$ for all $x \in V_c(p)^{\bot}$. 
If $\top p \Rightarrow \top \GA$ is provable, from Lemma \ref{disjOMR-},
there exists $\top C \in \top \GA$ such that $\top p \Rightarrow \top C$ is provable.
Then, $\GA \Rightarrow \DE, p$ is provable because of (/R), (cut), (wL), and the provability of $A \Rightarrow \top \top A$ and $\top \top A \Rightarrow A$.
Therefore, $\top p \Rightarrow \top \GA$ is not provable.
Then, from Lemma \ref{maxunp},  there exists an unprovable infinite sequent $\GA' \Rightarrow \DE'$ such that $\top p \in \GA'$, $\top \GA \subseteq \DE'$, and $\GA' \cup \DE' = \Omega_2$. This is a contradiction because $(\GA' \Rightarrow \DE') \not\perp_c (\GA \Rightarrow \DE, p)$ and $(\GA' \Rightarrow \DE') \in V_c(p)^{\bot}$.

\item $W_c \in \sigma_c$, $\emptyset \in \sigma_c$, if $X \in \sigma_c$, then $X^{\perp} \in \sigma_c$, and if $X \in \sigma_c$ and $Y \in \sigma_c$, then $X \cap Y \in \sigma_c$.

By Lemma \ref{pOMm}, $\sigma_c$ coincides with $\{ V_c(A) | A \in \Omega_2 \}$.
$V_c(\top) = W$, $V_c(p p \top p) = \emptyset$, if $V_c(A) = X$, then $V_c(\top A) = X^{\perp}$,
and if $V_c(A) = X$ and $V_c(B) = Y$, then $V_c(AAB) = X \cap Y$.

\item Orthomodularity

Because $AAB \Rightarrow B$ is provable by ($\wedge$L1), for all $\GA \Rightarrow \DE \in W_c$,
if $AAB \in \GA$, then $B \in \GA$.
Therefore, from Lemma \ref{pOMm}, for all $A,B \in \Omega_2$, if $(\GA \Rightarrow \DE) \models AAB $, then $(\GA \Rightarrow \DE) \models B$.
This implies that $V_c(AAB) (= V_c(A) \cap (V_c(A) \cap (V_c(A) \cap V_c(B))^{\bot})^{\bot}) \subseteq V_c(B)$. Therefore, $V_c(A) \cap (V_c(A) \cap (V_c(A) \cap V_c(B))^{\bot})^{\bot} \subseteq V_c(A) \cap V_c(B)$. An analogous argument applied to $A,B \Rightarrow AAB$ yields $V_c(A) \cap V_c(B) \subseteq V_c(A) \cap (V_c(A) \cap (V_c(A) \cap V_c(B))^{\bot})^{\bot}$.
\end{description}
\end{proof}

\begin{theorem}[Completeness theorem for {\bf OMR$^{\top}$}]\label{compOMRp} If $\GA \Rightarrow \DE$ (consisting of $\mathcal{L}_2$) is valid in all OM-models, then 
$\GA \Rightarrow \DE$ is provable in {\bf OMR$^{\top}$}.
\end{theorem}
\begin{proof}
We prove the contrapositive.
Suppose $\GA \Rightarrow \DE$ is not provable. Then, from Lemma \ref{maxunp}, 
there exists an unprovable infinite sequent $\GA' \Rightarrow \DE'$ such that $\GA \subseteq \GA'$, $\DE \subseteq \DE'$, and $(\GA' \Rightarrow \DE') \in W_c$.
By Lemma \ref{pOMm} and \ref{cOMm}, $\GA \Rightarrow \DE$ is false at $(\GA' \Rightarrow \DE') \in W_c$ in the OM-model $(W_c, \perp_c, \sigma_c, V_c )$.
\end{proof}

\begin{theorem}[Completeness theorem for {\bf OMR$^{S}$}]\label{compOMRS} If $\GA \Rightarrow \DE$ (consisting of $\mathcal{L}_2$) is valid in all OM-models, then 
$\GA \Rightarrow \DE$ is provable in {\bf OMR$^{S}$}.
\end{theorem}
\begin{proof}
A corollary of Lemma \ref{=Stop} and Theorem \ref{compOMRp}.
\end{proof}

\section{Conclusion and Future Work}

In this study, we investigated restricted negation as a primitive
connective for OML and developed a proof system based on
it. We showed that restricted negation can be understood as a relative
form of orthogonal complementation, that conjunction and ordinary
negation are definable from it together with a constant, and that the
orthomodular law has a simple formulation in this language.
The proposed proof system gives a new proof-theoretic perspective on
OML. Instead of treating orthomodularity as an additional rule, it
incorporates orthomodular behavior into the rules for restricted negation
itself. In this sense, the system reformulates OML around a single
connective that reflects its orthomodular structure.
These results may serve as a basis for a more detailed analysis of restricted negation.

Future work includes a detailed comparison with existing proof systems
for OML and the development of an extension admitting cut elimination.
Although the sequent calculus constructed in this study does not admit cut elimination, this is not a peculiarity of the present framework; rather, it reflects a general difficulty in OML.
Although several proof systems admitting cut elimination have been studied for OL, corresponding results for OML remain limited.
Structured sequent formalisms, such as hypersequents or nested sequents \cite{kashima1994,Poggiolesi2009},
may provide a promising direction for this purpose.

\clearpage

\bibliographystyle{eptcs}
\bibliography{generic}

\begin{thebibliography}{10}
\providecommand{\bibitemdeclare}[2]{}
\providecommand{\surnamestart}{}
\providecommand{\surnameend}{}
\providecommand{\urlprefix}{Available at }
\providecommand{\url}[1]{\texttt{#1}}
\providecommand{\href}[2]{\texttt{#2}}
\providecommand{\urlalt}[2]{\href{#1}{#2}}
\providecommand{\doi}[1]{doi:\urlalt{https://doi.org/#1}{#1}}
\providecommand{\eprint}[1]{arXiv:\urlalt{https://arxiv.org/abs/#1}{#1}}
\providecommand{\bibinfo}[2]{#2}

\bibitemdeclare{article}{bat2}
\bibitem{bat2}
\bibinfo{author}{A.~\surnamestart Baltag\surnameend} \&
  \bibinfo{author}{S.~\surnamestart Smets\surnameend} (\bibinfo{year}{2012}):
  \emph{\bibinfo{title}{The Dynamic Turn in Quantum Logic}}.
\newblock {\slshape \bibinfo{journal}{Synthese}}
  \bibinfo{volume}{186}(\bibinfo{number}{3}), pp. \bibinfo{pages}{753--773},
  \doi{10.1007/s11229-011-9915-7}.

\bibitemdeclare{article}{birk}
\bibitem{birk}
\bibinfo{author}{G.~\surnamestart Birkhoff\surnameend} \&
  \bibinfo{author}{J.~\surnamestart von Neumann\surnameend}
  (\bibinfo{year}{1936}): \emph{\bibinfo{title}{The Logic of Quantum
  Mechanics}}.
\newblock {\slshape \bibinfo{journal}{The Annals of Mathematics}}
  \bibinfo{volume}{37}(\bibinfo{number}{4}), pp. \bibinfo{pages}{823--843},
  \doi{10.2307/1968621}.

\bibitemdeclare{article}{OBrunet}
\bibitem{OBrunet}
\bibinfo{author}{O.~\surnamestart Brunet\surnameend} (\bibinfo{year}{2005}):
  \emph{\bibinfo{title}{A Rule-Based Logic for Quantum Information}}.
\newblock {\slshape \bibinfo{journal}{arXiv:cs/0504018}},
  \doi{10.48550/arXiv.cs/0504018}.

\bibitemdeclare{article}{Chajda}
\bibitem{Chajda}
\bibinfo{author}{I.~\surnamestart Chajda\surnameend} (\bibinfo{year}{2005}):
  \emph{\bibinfo{title}{Sheffer operation in ortholattices}}.
\newblock {\slshape \bibinfo{journal}{Acta Universitatis Palackianae
  Olomucensis. Facultas Rerum Naturalium. Mathematica}}
  \bibinfo{volume}{44(1)}, pp. \bibinfo{pages}{19--23}.

\bibitemdeclare{incollection}{HANDQ}
\bibitem{HANDQ}
\bibinfo{author}{M.~L.~Dalla \surnamestart Chiara\surnameend} \&
  \bibinfo{author}{R.~\surnamestart Giuntini\surnameend}
  (\bibinfo{year}{2002}): \emph{\bibinfo{title}{Quantum Logics}}.
\newblock In \bibinfo{editor}{D.~M. \surnamestart Gabbay\surnameend} \&
  \bibinfo{editor}{F.~\surnamestart Guenthner\surnameend}, editors: {\slshape
  \bibinfo{booktitle}{Handbook of Philosophical Logic 2nd Edition}},
  \bibinfo{volume}{6}, \bibinfo{publisher}{Springer, Dordrecht}, pp.
  \bibinfo{pages}{129--228}, \doi{10.1007/978-94-017-0460-1_2}.

\bibitemdeclare{article}{uwe}
\bibitem{uwe}
\bibinfo{author}{U.~\surnamestart Egly\surnameend} \&
  \bibinfo{author}{H.~\surnamestart Tompits\surnameend} (\bibinfo{year}{2003}):
  \emph{\bibinfo{title}{On Different Proof-Search Strategies for Orthologic}}.
\newblock {\slshape \bibinfo{journal}{Studia Logica}} \bibinfo{volume}{73}, pp.
  \bibinfo{pages}{131--152}, \doi{10.1023/A:1022993408070}.

\bibitemdeclare{book}{HANDQS}
\bibitem{HANDQS}
\bibinfo{author}{K.~\surnamestart Engesser\surnameend},
  \bibinfo{author}{D.~\surnamestart Gabbay\surnameend} \&
  \bibinfo{author}{D.~\surnamestart Lehmann\surnameend} (\bibinfo{year}{2008}):
  \emph{\bibinfo{title}{Handbook of Quantum Logic and Quantum Structures:
  Quantum Logic}}.
\newblock \bibinfo{publisher}{Elsevier}, \doi{10.1016/C2009-0-16935-0}.

\bibitemdeclare{article}{fag}
\bibitem{fag}
\bibinfo{author}{C.~\surnamestart Faggian\surnameend} \&
  \bibinfo{author}{G.~\surnamestart Sambin\surnameend} (\bibinfo{year}{1998}):
  \emph{\bibinfo{title}{From Basic Logic to Quantum Logics with
  Cut-Elimination}}.
\newblock {\slshape \bibinfo{journal}{International Journal of Theoretical
  Physics}} \bibinfo{volume}{37(1)}, pp. \bibinfo{pages}{31--37},
  \doi{10.1023/A:1026652903971}.

\bibitemdeclare{article}{Fazio2023}
\bibitem{Fazio2023}
\bibinfo{author}{D.~\surnamestart Fazio\surnameend},
  \bibinfo{author}{A.~\surnamestart Ledda\surnameend},
  \bibinfo{author}{F.~\surnamestart Paoli\surnameend} \&
  \bibinfo{author}{G.~St. \surnamestart John\surnameend}
  (\bibinfo{year}{2023}): \emph{\bibinfo{title}{A Substructural Gentzen
  Calculus for Orthomodular Quantum Logic}}.
\newblock {\slshape \bibinfo{journal}{Review of Symbolic Logic}}
  \bibinfo{volume}{16}(\bibinfo{number}{4}), pp. \bibinfo{pages}{1177--1198},
  \doi{10.1017/S1755020322000016}.

\bibitemdeclare{article}{fin}
\bibitem{fin}
\bibinfo{author}{P.~D. \surnamestart Finch\surnameend} (\bibinfo{year}{1970}):
  \emph{\bibinfo{title}{Quantum Logic as an Implication Algebra}}.
\newblock {\slshape \bibinfo{journal}{Bulletin of the Australian Mathematical
  Society}} \bibinfo{volume}{2(1)}, pp. \bibinfo{pages}{101--106},
  \doi{10.1017/S0004972700041642}.

\bibitemdeclare{article}{har}
\bibitem{har}
\bibinfo{author}{G.~M. \surnamestart Hardegree\surnameend}
  (\bibinfo{year}{1981}): \emph{\bibinfo{title}{Material Implication in
  Orthomodular (and Boolean) Lattices}}.
\newblock {\slshape \bibinfo{journal}{Notre Dame Journal of Formal Logic}}
  \bibinfo{volume}{22(2)}, pp. \bibinfo{pages}{163--182},
  \doi{10.1305/ndjfl/1093883401}.

\bibitemdeclare{article}{kashima1994}
\bibitem{kashima1994}
\bibinfo{author}{R.~\surnamestart Kashima\surnameend} (\bibinfo{year}{1994}):
  \emph{\bibinfo{title}{Cut-free Sequent Calculi for some Tense Logics}}.
\newblock {\slshape \bibinfo{journal}{Studia Logica}} \bibinfo{volume}{53}, pp.
  \bibinfo{pages}{119--135}, \doi{10.1007/BF01053026}.

\bibitemdeclare{article}{TK1}
\bibitem{TK1}
\bibinfo{author}{T.~\surnamestart Kawano\surnameend} (\bibinfo{year}{2018}):
  \emph{\bibinfo{title}{Advanced Kripke Frame for Quantum Logic}}.
\newblock {\slshape \bibinfo{journal}{Logic, Language, Information, and
  Computation. WoLLIC 2018. Lecture Notes in Computer Science}}
  \bibinfo{volume}{10944}, pp. \bibinfo{pages}{237--249},
  \doi{10.1007/978-3-662-57669-4_14}.

\bibitemdeclare{article}{tom3}
\bibitem{tom3}
\bibinfo{author}{T.~\surnamestart Kawano\surnameend} (\bibinfo{year}{2018}):
  \emph{\bibinfo{title}{Labeled Sequent Calculus for Orthologic}}.
\newblock {\slshape \bibinfo{journal}{Bulletin of the Section of Logic}}
  \bibinfo{volume}{47}(\bibinfo{number}{4}), pp. \bibinfo{pages}{217--232},
  \doi{10.18778/0138-0680.47.4.01}.

\bibitemdeclare{article}{TKMB}
\bibitem{TKMB}
\bibinfo{author}{T.~\surnamestart Kawano\surnameend} (\bibinfo{year}{2024}):
  \emph{\bibinfo{title}{Nested-Sequent Calculus for Modal Logic MB}}.
\newblock {\slshape \bibinfo{journal}{Electronic Proceedings in Theoretical
  Computer Science}} \bibinfo{volume}{415}, pp. \bibinfo{pages}{33--47},
  \doi{10.48550/arXiv.2501.00484}.

\bibitemdeclare{article}{Kornell2024}
\bibitem{Kornell2024}
\bibinfo{author}{A.~\surnamestart Kornell\surnameend} (\bibinfo{year}{2024}):
  \emph{\bibinfo{title}{A Natural Deduction System for Orthomodular Logic}}.
\newblock {\slshape \bibinfo{journal}{Review of Symbolic Logic}}
  \bibinfo{volume}{17(3)}, pp. \bibinfo{pages}{910--949},
  \doi{10.1017/S1755020323000229}.

\bibitemdeclare{article}{kotas}
\bibitem{kotas}
\bibinfo{author}{J.~\surnamestart Kotas\surnameend} (\bibinfo{year}{1967}):
  \emph{\bibinfo{title}{An Axiom System for the Modular Logic}}.
\newblock {\slshape \bibinfo{journal}{Studia Logica}} \bibinfo{volume}{21}, pp.
  \bibinfo{pages}{17--37}, \doi{10.1007/BF02123412}.

\bibitemdeclare{article}{McCune}
\bibitem{McCune}
\bibinfo{author}{W.~\surnamestart McCune\surnameend},
  \bibinfo{author}{R.~\surnamestart Padmanabhan\surnameend},
  \bibinfo{author}{M.~A. \surnamestart Rose\surnameend} \&
  \bibinfo{author}{R.~\surnamestart Veroff\surnameend} (\bibinfo{year}{2005}):
  \emph{\bibinfo{title}{Automated Discovery of Single Axioms for
  Ortholattices}}.
\newblock {\slshape \bibinfo{journal}{Algebra Universalis}}
  \bibinfo{volume}{52}, pp. \bibinfo{pages}{541--549},
  \doi{10.1007/s00012-004-1902-0}.

\bibitemdeclare{article}{nis1}
\bibitem{nis1}
\bibinfo{author}{H.~\surnamestart Nishimura\surnameend} (\bibinfo{year}{1980}):
  \emph{\bibinfo{title}{Sequential Method in Quantum Logic}}.
\newblock {\slshape \bibinfo{journal}{The Journal of Symbolic Logic}}
  \bibinfo{volume}{45}(\bibinfo{number}{2}), pp. \bibinfo{pages}{339--352},
  \doi{10.2307/2273194}.

\bibitemdeclare{article}{nis2}
\bibitem{nis2}
\bibinfo{author}{H.~\surnamestart Nishimura\surnameend} (\bibinfo{year}{1994}):
  \emph{\bibinfo{title}{Proof Theory for Minimal Quantum Logic I}}.
\newblock {\slshape \bibinfo{journal}{International Journal of Theoretical
  Physics}} \bibinfo{volume}{33}, pp. \bibinfo{pages}{103--113},
  \doi{10.1007/BF00671616}.

\bibitemdeclare{article}{ozawa2}
\bibitem{ozawa2}
\bibinfo{author}{M.~\surnamestart Ozawa\surnameend} (\bibinfo{year}{2017}):
  \emph{\bibinfo{title}{Orthomodular-Valued Models for Quantum Set Theory}}.
\newblock {\slshape \bibinfo{journal}{The Review of Symbolic Logic}}
  \bibinfo{volume}{10}, pp. \bibinfo{pages}{782--807},
  \doi{10.1017/S1755020317000120}.

\bibitemdeclare{article}{ozawa1}
\bibitem{ozawa1}
\bibinfo{author}{M.~\surnamestart Ozawa\surnameend} (\bibinfo{year}{2021}):
  \emph{\bibinfo{title}{Quantum Set Theory: Transfer Principle and De Morgan's
  Laws}}.
\newblock {\slshape \bibinfo{journal}{Annals of Pure and Applied Logic}}
  \bibinfo{volume}{172(4) 102938}, \doi{10.1016/j.apal.2020.102938}.

\bibitemdeclare{incollection}{Poggiolesi2009}
\bibitem{Poggiolesi2009}
\bibinfo{author}{F.~\surnamestart Poggiolesi\surnameend}
  (\bibinfo{year}{2009}): \emph{\bibinfo{title}{The Method of
  Tree-Hypersequents for Modal Propositional Logic}}.
\newblock In \bibinfo{editor}{D.~\surnamestart Makinson\surnameend},
  \bibinfo{editor}{J.~\surnamestart Malinowski\surnameend} \&
  \bibinfo{editor}{H.~\surnamestart Wansing\surnameend}, editors: {\slshape
  \bibinfo{booktitle}{Trends in Logic: Towards Mathematical Philsophy}},
  \bibinfo{publisher}{Springer}, pp. \bibinfo{pages}{31--51},
  \doi{10.1007/978-1-4020-9084-4_3}.

\bibitemdeclare{article}{samb}
\bibitem{samb}
\bibinfo{author}{G.~\surnamestart Sambin\surnameend},
  \bibinfo{author}{G.~\surnamestart Battilotti\surnameend} \&
  \bibinfo{author}{C.~\surnamestart Faggian\surnameend} (\bibinfo{year}{2000}):
  \emph{\bibinfo{title}{Basic logic: reflection, symmetry, visibility}}.
\newblock {\slshape \bibinfo{journal}{The Journal of Symbolic Logic}}
  \bibinfo{volume}{65}(\bibinfo{number}{3}), pp. \bibinfo{pages}{979--1013},
  \doi{10.2307/2586685}.

\bibitemdeclare{article}{sasaki}
\bibitem{sasaki}
\bibinfo{author}{U.~\surnamestart Sasaki\surnameend} (\bibinfo{year}{1954}):
  \emph{\bibinfo{title}{Orthocomplemented Lattices Satisfying the Exchange
  Axiom}}.
\newblock {\slshape \bibinfo{journal}{Journal of Science of the Hiroshima
  University}} \bibinfo{volume}{A 17}, pp. \bibinfo{pages}{293--302},
  \doi{10.32917/hmj/1557281141}.

\end{thebibliography}
\end{document}